\documentclass[10pt, conference, compsocconf]{IEEEtran}

\usepackage[latin1]{inputenc}
\usepackage{amsmath}
\usepackage{amsfonts}
\usepackage{amssymb}
\usepackage{times}
\usepackage{graphicx}
\usepackage{latexsym}
\usepackage{epsfig}
\usepackage{booktabs}
\usepackage{algorithmic} 
\usepackage{algorithm}

\newcommand{\cE}{\mathcal{E}}
\newcommand{\cR}{\mathcal{R}}

\newcommand{\cG}{\mathcal{G}}

\begin{document}

\newtheorem{definition}{Definition}
\newtheorem{theorem}{Theorem}
\newtheorem{remark}{Remark}

\title{Interesting Multi-Relational Patterns}

\author{Eirini Spyropoulou, \\
Department of Engineering Mathematics,\\
University of Bristol, UK\\
Email: eirini.spyropoulou@bristol.ac.uk \and
Tijl De Bie, \\
Department of Engineering Mathematics,\\
University of Bristol, UK\\
Email: tijl.debie@bristol.ac.uk
}

\maketitle

\begin{abstract}
Mining patterns from multi-relational data is a problem attracting increasing interest within the data mining community. Traditional data mining approaches are typically developed for highly simplified types of data, such as an attribute-value table or a binary database, such that those methods are not directly applicable to multi-relational data. Nevertheless, multi-relational data is a more truthful and therefore often also a more powerful representation of reality. Mining patterns of a suitably expressive syntax directly from this representation, is thus a research problem of great importance.

In this paper we introduce a novel approach to mining patterns in multi-relational data. We propose a new syntax for multi-relational patterns as complete connected subgraphs in a representation of the database as a $K$-partite graph. We show how this pattern syntax is generally applicable to multi-relational data, while it reduces to well-known tiles \cite{tiling} when the data is a simple binary or attribute-value table. We propose RMiner, an efficient algorithm to mine such patterns, and we introduce a method for quantifying their interestingness when contrasted with prior information of the data miner. Finally, we illustrate the usefulness of our approach by discussing results on real-world and synthetic databases.
\end{abstract}

\section{Introduction}

Since the formalization of frequent itemset mining and association rule mining, the focus of data mining research has mostly been on single-table databases. However, as most information systems rely on a multi-relational representation of data, the focus has recently started to shift to multi-relational databases (MRDs). On top of the challenges faced in most pattern mining research, a key additional challenge here is the definition of insightful pattern types that properly exploit or elucidate the structure in the data.

Previous work has focused on generalizing ideas from frequent pattern mining. The most common strategy is to first take the full join of all the tables of the MRD, after which standard pattern mining methods can be applied. However, in flattening the MRD in this way important structural information is inevitably lost \cite{starassociations,koopman1,Goethals}. Approaches relying on Inductive Logic Programming type of patterns avoid this, thus capturing better the structure of the MRD \cite{warmr,farmer,koopman2}. However, all these approaches rely on transferring the notions of \emph{recurring pattern} and \emph{support} in the multi-relational setting either by
measuring the support with respect to the entries of the join table \cite{starassociations,koopman1} or with respect to just one table or entity in the database \cite{Goethals,warmr,farmer,koopman2}, which complicates the interpretation of the results. (See Sec.~\ref{relwork} for a more detailed discussion of existing work.)

On top of these conceptual problems, most existing methods for mining MRDs also suffer from usability problems: the returned set of patterns is often overwhelmingly large and redundant, or subjectively not very interesting. Fortunately, these problems have recently been addressed by the pattern mining research community, albeit in simpler settings (mostly itemsets in binary databases). This includes the definition of new objective interestingness measures with various properties (see \cite{intmeasures} for an overview), as well as the definition of general schemes to formalize subjective interestingness \cite{swap07,idn,DeB:10a,tijl}. Another related development, mostly aimed at reducing redundancies, is the focus on evaluating interestingness of pattern sets, instead of individual patterns \cite{Item06Siebes,DeZ:07}. To improve multi-relational data mining methods, some of these ideas should be transferred and adapted where needed.

Here we contribute on both these fronts: the conceptualization and search for patterns in MRDs, and the quantification of their interestingness. In particular, in Sec.~\ref{syntax} we propose a new type of pattern syntax in MRDs that captures the structural information of an MRD. It does not rely on the concept of support, thus avoiding some of the pitfalls in earlier work on this topic. We represent the MRD as a $K$-partite graph and define a pattern as a complete connected subgraph in this $K$-partite graph. We illustrate that this type of pattern is \emph{easy to interpret}, it is \emph{generally applicable} to MRDs, while in simple settings it \emph{subsumes itemsets as a special case} (or more accurately, tiles \cite{tiling}). We further propose RMiner, \emph{an efficient algorithm} to mine such patterns directly from the $K$-partite graph representing the MRD (Sec.~\ref{algorithm}). In Sec.~\ref{assessment} we show that the proposed pattern syntax lends itself well to \emph{formalizating their subjective interestingness}, subject to certain prior knowledge on the data. In a similar way as the work in \cite{tijl} has done for itemsets in binary databases, this approach guarantees the interestingness of the returned patterns in a well-defined setting. We discuss related work in Sec.~\ref{relwork}. Finally, in Sec.~\ref{experiments} we show results on real-world and synthetic MRDs, to support the above claims.

\section{Multi-relational data and patterns}\label{syntax}

We first formalize multi-relational databases as considered in this paper. In an abstract manner this formalization is reminiscent of the Entity-Relationship (ER) model as explained in \cite{FODS}. Then we show how such an MRD can be uniquely represented as a $K$-partite graph. Finally, we move on to defining the proposed pattern syntax, based on the graph-representation of the MRD.

\paragraph{Multi-relational database (MRD)} We formalize a multi-relational database as a collection of \emph{entities} $e_k^i$ that are grouped into $K$ \emph{entity types} $E_k$ ($k=1:K$). Each entity type has a \emph{domain}, denoted as $\cE_k$ for entity type $E_k$, such that $e_k^i\in\cE_k$. For the purpose of this paper we assume domains are discrete. A \emph{relationship type} $R_{kl}$ between a pair of distinct entity types $E_k$ and $E_l$ defines a \emph{relationship set} $\cR_{kl}$ which contains \emph{relationship instances} $r_{kl}=(e_k,e_l)\in\cR_{kl}$ between pairs of entities $e_k\in\cE_k$ and $e_l\in\cE_l$. Relationship types can be many-to-many, one-to-many, or one-to-one, depending on how many relations the elements of either domains can participate in.

Let us consider a toy example with `year', `movie', and `genre' as entity types (with obvious domains), and with relationship types between `year' and `movie', specifying the year of release of the movie, and between `movie' and `genre', specifying the genres of a movie.  The first of these relations is a one-to-many relationship type, while the second is a many-to-many relationship type.

\begin{remark}
[What about attributes?]
In an ER model, an entity can have attributes associated to it. Our formalism differs from this, in treating each attribute as an entity type of its own, with the set of possible attribute values as its domain. Then, associating attribute values with the entity they correspond to is done by making use of a one-to-many relationship type between the entity type of the attribute and the entity \cite{FODS}. E.g., in the toy example considered before, `year' would typically be modelled as an attribute to the `movie' entity. However, we model it as an entity, with a relationship type between `year' and `movie'.
While this approach is inadequate for data modelling purposes, it allows for a unified treatment of attributes and entities. This is desirable, as in the ER model the distinction between attributes and entities is often ambiguous, while we wish our methods to be independent of such modelling choices. Furthermore, it makes our methods more general than other methods that do distinguish entities from attributes.
\end{remark}

\paragraph{A graph representation of a MRD}
In the rest of this paper, we will make use of a graph representation of MRDs. In this representation, there is a node for each entity in the MRD, and an edge between the nodes corresponding to entities $e_k$ and  $e_l$ for each relationship instance $r_{kl}=(e_k,e_l)$. We say that nodes representing entities of the same type are of the same node type, and similarly we say that edges representing relationship instances of the same type are of the same edge type. Clearly, the resulting graph is $K$-partite, each partition in the graph containing nodes of the same node type.

Because of this symmetry between entities/relations and nodes/edges, we slightly overload notation and denote the resulting graph as $\cG=(\cE,\cR)$ where $\cE=\bigcup\cE_k$ (all nodes of all types) and $\cR=\bigcup\cR_{kl}$ (all edges of all types). We will also refer to $E_k$ as the $k$'th node type, and to $R_{kl}$ as an edge type between node types $E_k$ and $E_l$.

As an example, the graph representation of the toy MRD is shown in Fig.~\ref{dbex}. In this example there are three entity types, namely `title', `year' and `genre'. Moreover there are two relationship types, one one-to-many relationship type between `year' and `title' and one many-to-many relationship type between `title' and `genre'.

\begin{figure}[!ht] 
\centerline{\includegraphics[height=2.5cm]{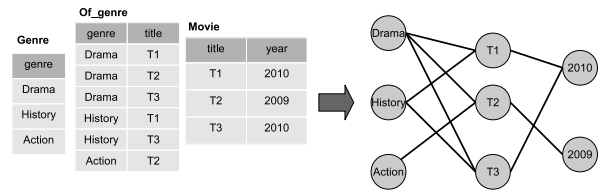}}
\caption{Database transformed to a $K$-partite graph. The entity types `genre', `title', `year' correspond to different parts in the graph and the entities of each entity type correspond to different nodes. The join table `Of\_genre' defines a many-to-many relationship between the entity types `genre' and `title' and the table `Film' defines an one-to-many relationship between entities `title' and `year'. Two entities are linked with an edge if they co-occur in a same tuple.}\center
\label{dbex}
\end{figure}


\paragraph{The pattern syntax}
In this paper, we are interested in identifying patterns that relate entities with each other, within and between different entity types. The pattern syntax we suggest for this purpose is a connected complete subgraph (CCS), and in particular a Maximal CCS (MCCS), in the $K$-partite graph representation of the MRD.

\begin{definition}[Maximal connected complete subgraph] Given a graph $\cG=(\cE,\cR)$ with $\cE=\bigcup\cE_k$ and $\cR=\bigcup\cR_{kl}$, a subgraph is defined as a graph $(\cE',\cR')$ for which $\cE'\subseteq\cE$ and $\cR'\subseteq\cR$ such that for every pair of nodes $e_k,e_l\in\cE'$ part of the subgraph, if $(e_k,e_l)\in\cR$ then $(e_k,e_l)\in\cR'$. A subgraph $(\cE',\cR')$ is connected if there exists a path between any pair of nodes from $\cE'$ along edges from $\cR'$. It is complete if for any pair of nodes $e_k,e_l\in\cE'$ of different types (i.e. with $e_k\in\cE_k$ and $e_l\in\cE_l$, $k\neq l$) between which an edge type $R_{kl}$ exists, it holds that $(e_k,e_l)\in\cR'$. A maximal connected complete subgraph is a connected complete subgraph to which no node can be added without violating connectedness or completeness.
\end{definition}
Note that a subgraph of size larger than one can be connected only if it contains nodes of at least two different types. A connected complete subgraph is a generalization of a clique to the $K$-partite graph representation of the MRD, used in this paper. 

In the example of Fig.~\ref{dbex} the set of nodes $\{$T1, T3, Drama, History, 2010$\}$ represent an MCCS pattern. This pattern provides the information that titles T1 and T3 are both produced in 2010 that are both Drama and History.

\paragraph{Special cases of MCCSs}
Conceptually, MCCSs are easy to grasp, and the empirical results will further demonstate that this pattern syntax is a sensible and intuitive one. An additional argument in support of MCCSs is that they reduce to well-known pattern syntaxes of well-studied forms of data.

Consider a market-basket database, containing two entity types: items and transactions. There is one relationship type representing the fact that an item was bought in a transaction. It is well-known that a binary item-transaction database can be represented by means of a bipartite graph \cite{zaki98}, and indeed this graph is exactly the graph representation of this rather degenerate case of a MRD. An MCCS in this bipartite graph is a maximal biclique, which corresponds to a maximal tile in this database: the pair of a closed itemset and its supporting transactions \cite{tiling}. This is depicted in Fig.~\ref{mbex}, showing a database of three items and four transactions and the corresponding bipartite graph. The set of nodes $\{$T1, T2, I1, I2$\}$ is an example of a maximal biclique in this graph.

\begin{figure}
\centering
\centerline{\includegraphics[height=2.5cm]{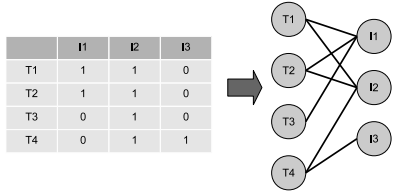}}
\caption{Transaction database as a bipartite graph. Transactions and items represent different partitions of the graph and are linked with edges according to the `1's of the binary matrix.}\center
\label{mbex}
\end{figure}
\begin{figure}
\centering
\centerline{\includegraphics[height=2.5cm]{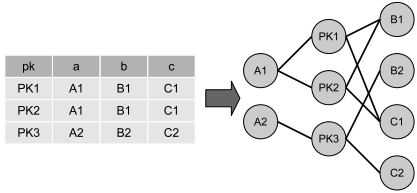}}
\caption{Attribute-value database as a k-partite graph. Attributes represent different partitions of the graph and attribute values represent the nodes. Key and non-key attribute values are linked if they co-occur in the same tuple of the attribute-value table.}\center
\label{atex}
\end{figure}

Similarly, for a single attribute-value data table the entity types in our formalization consist of the entity type that uniquely identifies the rows of the table (typically identified by a primary key attribute), along with an entity type for each of the (non-key) attributes. Hence, for an attribute-value table with $K-1$ (non-key) attributes, we would have $K$ entity types. An MCCS in this $K$-partite graph contains a set of nodes representing attribute-values and necessarily also a set of nodes corresponding to the rows of the table. Mapped back to the data table, this is equivalent to a set of attribute-values along with the supporting set of rows in the table. This is depicted in Fig.~\ref{atex} which shows an attribute-value data table with three attributes and three transactions. Here, the set of nodes $\{$A1, PK1, PK2, B1, C1$\}$ is an example of an MCCS.

\section{RMiner: an algorithm to search for all MCCSs}\label{algorithm}

In this section we outline a new algorithm called RMiner (from Relational Miner), that is able to efficiently enumerate all CCSs in a $K$-partite graph. After this, we present a variant of RMiner that can be used when only the Maximal CCSs are required.

Note that practical methods to enumerate all cliques in an ordinary graph exist (e.g. the Bron-Kerbosch algorithm \cite{bronkerbosch}). Furthermore, an algorithm for searching for $K$-partite maximal cliques was proposed in \cite{clicks}, although there edges are allowed between any pair of types. Despite clear similarities between these problem settings and the task of enumerating (M)CCSs in this paper, the differences are too large to allow the use of these algorithms to our problem without significant modifications or inefficiencies.

\paragraph{Enumerating all CCSs} In RMiner a CCS is represented using a list of the nodes $\mathbf{e}=\{e^1,e^2,\ldots,e^{|\mathbf{e}|}\}$ it contains. It enumerates CCSs by enumerating all such lists of nodes that represent a CCS. To do this efficiently, such lists are organized in a tree-structured search space, with the empty list (corresponding to the empty CCS) at the root and with a CCS represented by the node list $\mathbf{e}=\{e^1,e^2,\ldots,e^{|\mathbf{e}|-1},e^{|\mathbf{e}|}\}$ having its longest non-trivial prefix $\{e^1,e^2,\ldots,e^{|\mathbf{e}|-1}\}$ as its parent. RMiner traverses this tree in a depth-first manner, backtracking as soon as a subgraph is constructed that is no longer connected or complete. Note that this strategy ensures that each CCS is represented by a list for which it holds that each prefix is also a CCS. (For some permutations of the nodes, some prefixes may not be connected and thus would not represent a CCS.) We say such a representation has \emph{connected prefixes}.

To branch down in this search tree, we need an efficient way to identify which nodes may be added to a CCS such that the result is still a CCS. These nodes are the adjacent common neighbours to the CCS, defined as follows.

\begin{definition}[Common neighbour, adjacent common neighbour]
A node $e_k$ of type $E_k$ is a \emph{common neighbour} of a CCS represented by the list of nodes $\mathbf{e}$, iff for each node $e_l\in\mathbf{e}$ of type $E_l$ for which an edge type $R_{kl}$ exists between $E_k$ and $E_l$, it holds that $(e_k,e_l)\in\cR_{kl}$. An \emph{adjacent common neighbour} of a CCS is a common neighbour that is connected by an edge to at least one node from that CCS.
\end{definition}

This approach by itself does not rule out enumerating the same CCS more than once. Indeed, typically there are several permutations of the nodes in a CCS that have connected prefixes, such that the same CCS would be enumerated several times represented by different permutations of the node list. To avoid generating a CCS in more than one of these permutations, a lexicographical ordering is imposed over the nodes. RMiner is designed so that it only generates the lexicographically smallest permutation for each CCS, referred to as the \emph{representative permutation}. Importantly, the prefix of a representative permutation is a also representative permutation, such that the set of representative permutations forms a subtree of the larger search tree described above.

Note that limiting the search to representative permutations only, is not as easy to implement as for ordinary cliques in ordinary graphs (as done e.g. in the Bron-Kerbosch algorithm \cite{bronkerbosch}), since the nodes in a representative permutation are not necessarily sorted lexicographically themselves. Indeed, the lexicographically sorted list of nodes in a CCS may not have connected prefixes. Thus, representative permutations cannot be generated simply by adding nodes in lexicographical ordering. To specify how RMiner deals with this, we need to introduce the concept \emph{reachability}.

\begin{definition}[Reachable node type and reachable node]
A node type is said to be \emph{reachable} from a CCS iff there exists an edge type between this node type and the type of at least one of the nodes already in the CCS. We say a node is \emph{reachable} when its node type is reachable.
\end{definition}

RMiner will expand a CCS represented by $\mathbf{e}$ with an adjacent common neighbour $e$ only if $e$ is lexicographically larger than all nodes \emph{that were added since $e$ first became reachable}. This concludes the high-level description of RMiner for mining all CCSs. Correctness and efficiency of RMiner's enumeration strategy can be proven as follows.

\begin{theorem}[Correctness]
RMiner will enumerate each CCS at least once.
\end{theorem}
\begin{proof}
The algorithm would reconstruct any given CCS by first adding the lexicographically smallest node of the CCS, after which the lexicographically smallest \emph{reachable} node from the CCS is added, recursively until all of its nodes are added. That all nodes can be added follows from the fact that a CCS is connected by definition, such that each node type will become reachable at some point in the algorithm.
\end{proof}

\begin{theorem}[Efficiency]
RMiner will enumerate each CCS no more than once.
\end{theorem}
\begin{proof}
Assume the contrary, that the algorithm would generate at least two permutations representing a given CCS. Let us say that the first nodes in the permutations that differ are $e'$ and $e''$ respectively---i.e. all nodes preceding $e'$ in the first permutation and preceding $e''$ in the second permutation are equal. This implies that they are both reachable from the CCS comprised of the identical set of previously added nodes. Assume without loss of generality that $e'<e''$ in the lexicographical ordering. Then, in the second permutation, $e'$ was added after $e''$, which is not possible since $e''$ was reachable at the same time as $e'$ and can thus not be added any more after $e'$, showing contradiction.
\end{proof}

\paragraph{Restricting the search to MCCSs}
As the number of CCSs can still be prohibitively large, in this paper we are interested only in Maximal CCSs (MCCSs). Enumerating only the MCCSs can be done more efficiently than enumerating all CCSs, by pruning parts of the search tree that do not lead to Maximal CCSs. 

Our approach is based on the following observation (which is similar to the observation that allows pruning non-maximal partial solutions in the setting of frequent itemset mining \cite{charm}). Assume nodes $e'$ and $e''$ with $e'<e''$ are both extensions of a CCS. Then, if the common neighbours of the CCS extended by $e''$ are included in the set of common neighbours of the CCS extended by $e'$, any MCCS including that CCS as well as both $e'$ and $e''$ would be discovered when extending the CCS with $e'$. Extending it directly with $e''$ can only lead non-maximal CCSs since $e'$ cannot be included after including $e''$, so we can prune the branch of the search tree extending the CCS with $e''$.

Clearly, MCCSs can only be found whenever a recursion reaches an end point. That said, some of these end points may in fact not be maximal. Checking maximality can be done by verifying if the set of \emph{adjacent} common neighbours is equal to the CCS itself.

\begin{algorithm}
\caption{Simplified pseudocode for mining all MCCSs from a $K$-partite graph.}
\label{RMiner}
\textbf{RMiner}(Graph $\cG=(\cE,\cR)$)
\begin{algorithmic}[1]
\FOR{$e \in \cE$ in lexicographical order }
\STATE Expand$(\cG,e)$
\ENDFOR
\end{algorithmic}
\textbf{Expand}(Graph $\cG=(\cE,\cR)$, Ordered node list $\mathbf{e}\in\cE^n$)
\begin{algorithmic}[1]
\STATE $\mathcal{N}=\emptyset$
\FOR{$e \in \cE$ in lexicographical order }
\STATE $k=\max\{l:e \mbox{ is not reachable from the prefix } \mathbf{e}_{1:l}\}$
\IF{$\forall e'\in\mathbf{e}_{k+1:n}:e'<e$}
\STATE $\mathbf{n}=$ the set of common neighbours of $\mathbf{e}e$
\IF{$\nexists\mathbf{n}'\in\mathcal{N}:\mathbf{n}\subseteq \mathbf{n}'$}
\STATE $\mathcal{N} = \mathcal{N}\cup\{\mathbf{n}\}$
\STATE Expand($\cG$,$\mathbf{e}e$)
\ENDIF
\ENDIF
\ENDFOR
\IF{$\mathcal{N}\neq\emptyset$ \&\& IsMaximal($\cG$,$\mathbf{e}$)}
\STATE Send $\mathbf{e}$ to the output
\ENDIF
\end{algorithmic}
\textbf{IsMaximal}(Graph $\cG=(\cE,\cR)$, Ordered node list $\mathbf{e}\in\cE^n$)
\begin{algorithmic}[1]
\STATE $\mathbf{a} =$ set of adjacent common neighbours of $\mathbf{e}$ in $\cG$
\IF{$\mathbf{a}==\mathbf{e}$}
\STATE return true
\ELSE
\STATE return false
\ENDIF
\end{algorithmic}
\end{algorithm}

\paragraph{Implementation details}
Simplified pseudocode of our algorithm is given in Algorithm~\ref{RMiner}. For space and transparancy reasons, the pseudocode hides the following implementation details that allow for additional efficiency.

Each intermediate CCS, represented by a node list $\mathbf{e}$ in the pseudocode, is actually represented by RMiner in a manner similar to the Itemset-Tidset pairs in \cite{charm}, with some additional intricacies due the fact that items and transactions coincide when searching for MCCSs, and due to the graph being $K$-partite. More specifically, for each CCS, three pieces of information are kept in memory: the set of nodes $\mathbf{e}$ already in the CCS, the set of common neighbours of $\mathbf{e}$ (remembering also which are adjacent), and the subset of the adjacent common neighbours that are lexicographically larger than the last node added to $\mathbf{e}$ since their node type became reachable. These pieces of information can be stored efficiently, upon extension of a CCS they can be updated by means of simple set operations, and they facilitate pruning and maximality checking. Various further optimizations can be made (such as the use of diffsets), but they will be part of our future work. Another implementation detail is that the lexicographical order on the nodes is created by assuming an overriding lexicographical order over the node types, along with a lexicographical order over the nodes within each type.

\section{Assessment of patterns}\label{assessment}

The number of MCCSs is usually very large, which is a recurring problem in Pattern Mining research. Typically this problem is addressed by selecting or ranking patterns using objective or subjective interestingness measures \cite{intmeasures}. Here, we choose to define interestingness with respect to a specific type of prior information, by defining an interestingness measure which deems an MCCS to be more interesting if it is more unexpected given this prior information. More specifically, we consider as prior information the degree of each node in the different relationship types of the $K$-partite graph representation of the MRD. An MCCS is more interesting if it is harder to explain based on this prior information alone. For example in the setting of a movie MRD, an MCCS containing directors that have directed many movies would be deemed less interesting by our approach than an equally large MCCS containing less prolific directors, as the latter MCCS cannot as easily be attributed to randomness and is more unexpected.

To introduce the interestingness measure, we can closely follow the work presented in \cite{tijl,DeB:10a}, where it is argued that subjective interestingness can be formalized by contrasting patterns with a background model that is the Maximum Entropy model subject to the prior information. Thus we only need to detail the Maximum Entropy model for our case (see Sec.~\ref{maxent}), and the approach to contrast MCCS patterns with this model to arrive at an interestingness measure (see Sec.~\ref{contrast}).

\subsection{Maximum-Entropy model of the user's prior information}\label{maxent}

We consider as prior information 
the degree of the nodes for every relationship type in the $K$-partite graph representation of the MRD. Following \cite{tijl}, we formalize this prior information in a probability distribution $P$, 
fitting the Maximum Entropy distribution on the $K$-partite graph of the \emph{MRD}, with constraints on the expected degree of the nodes for every relationship type being equal to their actual degree. This is the distribution of maximal uncertainty about the data with only the prior information as bias.

The nature of the constraints is such that they are defined for every relationship type $R_{kl}$ of the \emph{MRD} without imposing any dependence between the relationship types.
Therefore, the Maximum Entropy distribution for the MRD subject to these constraints will be a product of independent Maximum Entropy distributions, one for each relationship type. Indeed, if there were dependencies between the relationship types, the Entropy of the joint distribution would be reduced by their mutual information \cite{CoverThomas}, and would therefore not be maximal. Representing each relationship type as a binary database $D_{kl}$ with $D_{kl}(i,j)=1$ when $(e_k^i,e_l^j)\in \cR_{kl}$, the Maximum Entropy distribution for the MRD is thus:
\begin{eqnarray*}
P(\cup_{kl}D_{kl})&=&\prod_{kl} P_{kl}(D_{kl}).
\end{eqnarray*} 

Maximizing the Entropy for every relationship type $R_{kl}$ of the \emph{MRD} represented by a binary matrix $D_{kl}$ subject to constraints on the expected degrees of the nodes is equivalent to maximising the Entropy of a distribution for a binary database subject to constraints on the expected row and column sums. The solution of this problem was shown to be a product of independent Bernoulli distributions, given by \cite{tijl}:


\begin{equation*}
\begin{array}{rcl}
P_{kl}(D_{kl}) &=& \prod\limits_{i,j}P_{kl}^{ij}(D_{kl}(i,j)) \mbox{ with } P_{kl}^{ij}(D_{kl}(i,j)),\\
               &=&\dfrac{\exp\left(D_{kl}(i,j)(-\lambda^i_{kl} -\mu^j_{kl})\right)}{1+\exp(-\lambda^i_{kl} -\mu^j_{kl})},
\end{array}
\end{equation*}
where $\lambda^i_{kl}$, $\mu^j_{kl}$ are parameters that can be computed efficiently.

\subsection{Contrasting MCCSs with the Maximum Entropy model}\label{contrast}

An interesting pattern conveys as much information as possible when contrasted with the user's prior information, as concisely as possible. Following earlier work \cite{tijl}, we can formalize this idea by quantifying the interestingness of an MCCS pattern $\gamma = (\cE',\cR')$ as the ratio of the self information of the MCCS and its description length:
\begin{eqnarray*}
\mbox{Interestingness}(\gamma)=\frac{\mbox{SelfInformation}(\gamma)}{\mbox{DescriptionLength}(\gamma)}.
\end{eqnarray*}

Here, the self information of an MCCS is defined given the probability of its edges under the Maximum Entropy model, as:

\begin{eqnarray*}
\mbox{SelfInformation}(\gamma)= - \sum_{r_{kl}^{ij}\in \cR'}\log(P^{ij}_{kl}(1)).
\end{eqnarray*}

An MCCS is described most naturally by the set of nodes it contains. More specifically, we choose to describe MCCS patterns by specifying for each node whether it does or does not belong to the pattern. To specify that a node belongs to an MCCS, we will use $-\log(p)$ bits, and to specify it does not belong to the MCCS we will use $-\log(1-p)$ bits, where $p$ is a probability parameter. Such a code satisfies Kraft's inequality exactly, and is thus optimal and asymptotically achievable \cite{CoverThomas}. Using this approach, the description length of an MCCS pattern $\gamma = (\cE',\cR')$ with $n=|\cE'|$ nodes and given that the graph of the MRD has $N=|\cE|$ nodes is given by:

\begin{equation*}
\begin{array}{rcl}
\mbox{DescriptionLength}(\gamma)&=&-\sum\limits_{i\in \gamma} \log(1-p) - \sum\limits_{i\not\in \gamma}\log(p),\\
&=& N\log\left(\frac{1-p}{p}\right) + n\log\left(\frac{1}{1-p}\right).
\end{array}
\end{equation*}

In \cite{tijl} it was suggested to set $p$ by default to the density of the database, an approach we adopted in our empirical results as well. However, the parameter can be tuned so as to bias the search more toward larger in number of nodes MCCSs (larger $p$) or toward smaller in number of nodes MCCSs (smaller
$p$), if desired.

\section{Related Work}\label{relwork}
Most previous work on relational pattern mining can be categorised into methods that generalize ideas from frequent itemset mining to the relational setting and methods that are based on Inductive Logic Programming (ILP). In this section we discuss the differences between these approaches and our approach as well as other works that do not fall into these two categories.

Well known ideas and algorithms from frequent itemset mining can be used for MRDs unaltered if applied on the join of all tables. The syntax of this type of patterns is essentially that of itemsets, with items in this case being attribute values and transactions being the tuples of the join table~\cite{starassociations,Goethals,koopman1}. The characteristic of this pattern syntax is that a tuple always contains one attribute value per attribute and as a result it is impossible to have two values of the same attribute in the same pattern. An itemset of this type for instance would not be able to capture the fact that a director can be related to many films. This is something that an MCCS pattern naturally captures. However, itemsets on the join table can still capture co-occurrences of attribute values that belong to different attributes.

On the other hand, the support, measured as the ratio of the tuples of the join table that contain an itemset, does not have a clear meaning as attribute values are replicated due to the join operation. A different approach is taken by Smurfig~\cite{Goethals} where the support is measured with respect to every table, as the relative number of keys that the items correspond to.

Warmr~\cite{warmr} and Farmer~\cite{farmer} are methods based on ILP. The patterns have the form of logic rules which can be regarded as local models of the database. The goal of these methods is to  mine for the most frequent rules. The support is defined as the relative number of key values of one target table that satisfy the rule. Therefore the more general the rule the higher its support will be. This type of pattern syntax is very expressive and can capture well the relational structure. However, the objective of these methods (frequent rules about the data) is different than ours (interesting patterns of co-occurring attributes). Finally the interestingness measure we propose in Sec.~\ref{assessment} can not be applied on Warmr and Farmer patterns and evaluating the interestingness of this kind of patterns is a challenge. 

Warmr, Farmer, and Smurfig are all based on the notion of a recurring pattern, and they directly depend on a support notion. Measuring the support with respect to one or a set of target tables, makes the results difficult to interpret and therefore introduces usability issues. The potential user will have to understand what exactly it means for a recurring pattern to be frequent with respect to a certain target table. Additionally, these techniques are likely to suffer from the same problems as other frequent pattern mining techniques, in particular the fact that support is usually only weakly related to interestingness. 

RDB-Krimp~\cite{koopman2} is a method for mining relational databases which is related to ours in that it also uses information theoretic ideas for the assessment of patterns. It uses the pattern syntax of Farmer~\cite{farmer} but considers just patterns of depth two (patterns of a target table and all the tables related to it with a foreign key). The most frequent patterns of this kind are mined for every table of the database as a target table and then RDB-Krimp finds the most characteristic patterns among them using the MDL principle. The focus of this method is on the total description length of the database joint with the patterns, and patterns are deemed more interesting if they are better at compressing this description length. We instead deem patterns more interesting if they describe surprising aspects of the database in a concise way, which we argue makes our results more relevant to an end-user. Finally RDB-Krimp relies on heuristic search to find the optimal set of patterns that best compress the database which is not the case for our method that searches exhaustively.

An approach for assessing the statistical significance of relational (SQL) queries based on randomisations of different tables is proposed in \cite{Ojala}. Although this approach was not intended to propose a method to mine such patterns it provides an insight towards making relational patterns useful to the user.

Finally one could see our work being connected to frequent sub-graph mining~\cite{gspan,frequentSubraphs} however besides being based on frequency, these methods are aimed at databases of many graphs rather than one connected graph.

\section{Empirical Results}\label{experiments}

To illustrate the kind of patterns retrieved by RMiner, we show and discuss empirical results on real world data. We additionally provide a comparison of our method with other methods in three different levels, namely qualitative comparison of the results, objective comparison of the rank of artificially embedded patterns in data and computational comparison. We compare with two representative methods of the previous approaches Smurfig~\cite{Goethals} and Farmer~\cite{farmer}. We chose to compare with Farmer rather than Warmr~\cite{warmr} because its pattern language is closer to that of MCCSs. Also note that RDB-Krimp was not publicly available.

\subsection{Data}
We performed experiments on three datasets of different size and complexity taken from the IMDB database (\emph{imdb-3ent-1year}, \emph{imdb-4ent-1year}, and \emph{imdb-4ent-3years}, see below for details)\footnote{See \texttt{http://www.imdb.com/}. Please note that there are some inconsistencies between the version of IMDB that we downloaded and the on-line version.} as well as on the Student database of the Computer Science department of the University of Antwerp \cite{Goethals} (called \emph{studentdb} in this paper). The purpose of the experiments in this section is to demonstrate computational feasibility, to illustrate the kind of patterns we get when mining real MRDs, and to show the usefulness of the interestingness measure when used to rank such patterns.

The \emph{imdb-3ent-1year} dataset contains directors and genres related to movies produced in 2010 and the \emph{imdb-4ent-1year} additionally contains keywords. The \emph{imdb-4ent-3years} dataset contains directors, genres and keywords, related to movies produced between 2008 and 2010. The Entity Relationship diagram of the IMDB derived databases is shown in Fig.~\ref{imdber}. The \emph{studentdb} dataset contains students, related to professors and courses. The Entity Relationship diagram of this database is shown in Fig.~\ref{smurfer}. General statistics about all the datasets are shown in Table~\ref{exp1}. 

\begin{table*}[!ht]
\label{datasets}
\begin{center}
{
\caption{Database details and computation times of RMiner.}\label{exp1}
\begin{tabular}{c c c c c c c}
\toprule
&Non-key Attributes & Tuples & Nodes& Edges & Patterns &Time(sec)\\
&per Table & per Table & & & & \\
&(excluding the join tables) & (excluding the join tables) & & & & \\
\midrule
imdb-3ent-1year& (1, 1, 1) & (15702, 14400, 28) & 30130& 48976& 14202& 9 \\
imdb-4ent-1year& (1, 1, 1, 1) & (15702, 14400, 28, 10878) & 31009& 80981& 4049& 379\\
imdb-4ent-3years& (1, 1, 1, 1) & (59656, 49348, 28, 32762) & 141796& 387577& 22393& 8084\\
studentdb& (3,1,1) & (154, 40, 174) & 401 & 3558 & 155 & 2\\
\bottomrule
\end{tabular}
}
\end{center}
\end{table*}

\begin{figure}
\centering
\centerline{\includegraphics[height=2cm]{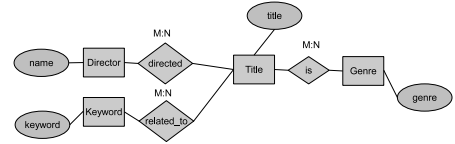}}
\caption{Entity-Relationship diagram of the \emph{imdb-4ent-1year} and \emph{imdb-4ent-3years} databases (\emph{keys} are omitted for clarity).}\label{imdber}\center
\end{figure}
\begin{figure}
\centering
\centerline{\includegraphics[height=2.5cm]{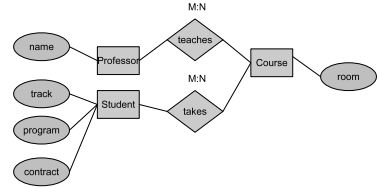}}
\caption{Entity-Relationship diagram of the \emph{studentdb} database (\emph{keys} are omitted for clarity).}\label{smurfer}\center
\end{figure}

\subsection{Results using RMiner}
We show results of RMiner on the \emph{imdb-3ent-1year} and \emph{studentdb} so that the resulting MCCS patterns are smaller and can be shown in this paper. However, we refer the reader to the website of this paper\footnote{https://sites.google.com/site/rminer2011/} for a full list of results from every dataset.

For all experiments we searched only for MCCSs that contain at least one node of each entity type because we wanted to show patterns that are different from tiles and thus truly relational. However we emphasize that our method works even without this constraint. Table~\ref{exp1} summarizes the output sizes and the computation times of the mining step.

\begin{figure*}[!ht] 
\centerline{\includegraphics[height=9cm]{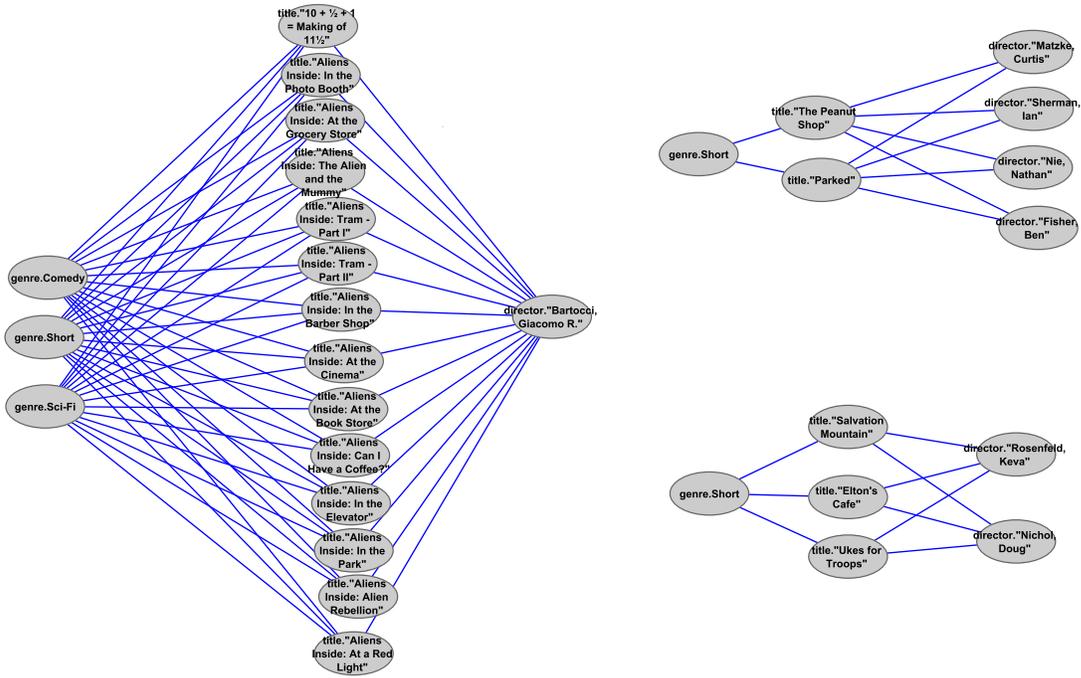}}
\caption{The three most interesting MCCS patterns in \emph{imdb-3ent-1year}. The leftmost MCCS is the first, the right top is the second and the right bottom MCCS is the third MCCS pattern.}\center
\label{results}
\end{figure*}

\paragraph{IMDB Database}
Figure~\ref{results} shows the top three patterns from the \emph{imdb-3ent-1year} database, ranked based on the interestingness measure defined in Section~\ref{assessment}. We expect patterns that convey as much information as possible as concisely as possible to be high in the list. These will be compact MCCS patterns containing many edges that are unlikely under the Maximum Entropy model.

The top-ranked pattern (leftmost in Fig.~\ref{results}) informs us about a director who directed fourteen movie titles that all of them are of genre Sci-Fi, Comedy and Short. More technically this pattern contains many edges and many of them are very unlikely edges under the model of prior information of the user, given that this director has just directed only these fourteen titles. 
The second pattern (top right in Fig.\ref{results}) shows that four directors directed two titles of the genre Short. The third pattern tells us about two directors that directed three films which are of the same genre. The directors included in second ranked pattern have directed only the titles of this pattern. The same holds for the third ranked pattern except for one director who has directed two more titles not included in this pattern. Hence the links between the directors and the titles in both these patterns are very unlikely, such that the information content of the MCCS is high. While explaining a relatively large number of unlikely edges, the number of nodes in the MCCS and thus its description length is relatively small. This explains the high interestingness.

\begin{figure*}[!ht] 
\centerline{\includegraphics[height=4cm]{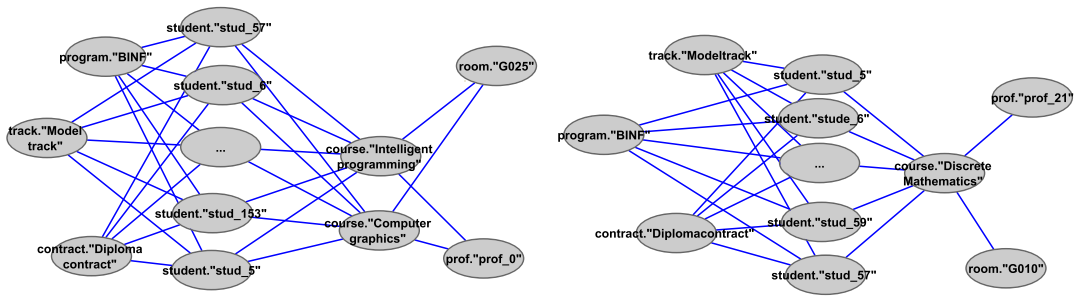}}
\caption{The first (left) and the third (right) most interesting MCCS patterns in the \emph{studentdb} dataset. Note that the number of student nodes is too large to show here, so we collapsed them onto one node labelled with an ellipsis.}\center
\label{results2}
\end{figure*}

\paragraph{Student Database dataset}
The top-ranked MCCSs on the \emph{studentdb} database are shown in Fig.~\ref{results2}. Since the first two patterns were structurally similar (although they convey non-redundant information), Fig.~\ref{results2} shows only the first and the third most interesting patterns but we discuss all the first three of them. 

The top two patterns convey information about a set of students, the program, contract, the track they are following, two courses they attend, the professor teaching these courses and the lecture room they are taught in. 
The difference between these two patterns is that the first contains a set of 67 students who follow the "Model" track while the second contains a set of 46 students who follow the "Individualised" track. 
Roughly speaking, the second pattern ranks lower than the first as it contains 126 fewer edges. Finally the third pattern (right in Fig.~\ref{results2}) conveys information about a set of 67 students, their program, contract, and track, as well as one course, the professor teaching it and the lecture room. The third pattern is less interesting than the first as it contains just 1 node less while it explains 67 fewer edges and contains is more common course room.

\subsection{Qualitative comparison} \label{qualitative} 
Here we qualitatively compare with the results of Smurfig~\cite{Goethals} and Farmer~\cite{farmer} on the \emph{imdb-3ent-1year} dataset.
\begin{figure}
\centering
\centerline{\includegraphics[height=0.7cm]{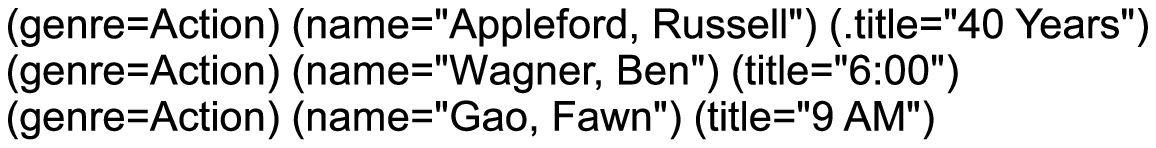}}
\caption{Three of the most frequent patterns of Smurfig on \emph{imdb-3ent-1years} database.}\label{smurfig}\center
\end{figure}

\begin{figure}
\centering
\centerline{\includegraphics[height=0.7cm]{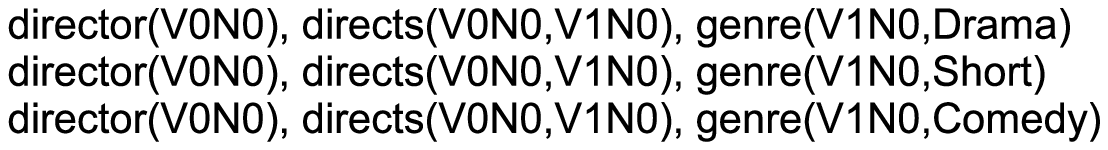}}
\caption{Top three most frequent patterns of Farmer on \emph{imdb-3ent-1years} database.}\label{farmer}\center
\end{figure}

\paragraph{Smurfig patterns}
We ran Smurfig with a support threshold of 0.001 to be as inclusive as possible. To compare with the patterns of RMiner we selected the ones that contain items from all the three attributes. As pointed out in Sec.~\ref{relwork}, each of these patterns can contain only one attribute value per attribute. Because of the nature of the \emph{imdb-3ent-1year} dataset each of them has absolute support of 1. Figure~\ref{smurfig} shows three of these patterns. Thus, Smurfig is clearly not suited to find relations in relational data of this kind.

\paragraph{Farmer patterns}
We ran Farmer with an absolute support threshold of 1. The pattern syntax we used had the following form: $director(X)$, $directs(X,Y)$, $genre(Y,g_1)\ldots genre(Y,g_n)$ and the key of the search is the atom $director(X)$.
Figure~\ref{farmer} shows the top three most frequent of these patterns that contain all three predicates. None of these patterns contain more than one genre constants, which is to be expected as the most frequent rules are bound to be the more general rules. Note that if we found the directors and titles that satisfy these rules, these patterns would correspond to CCSs. The difference between Farmer patterns and CCSs is analogous to the difference between itemsets and tiles. Farmer patterns corresponding to MCCSs are expected to be less frequent as they are more specific.

\subsection{Objective Comparison}
We investigated how different methods detect artificially embedded MCCS patterns of different sizes in the \emph{imdb-3ent-1year} data. More specifically, we investigated how highly the embedded MCCS (or a larger MCCS containing it) is ranked by our method using the interestingness measure we propose. To compare with Farmer we checked the rank of the most frequent rule corresponding to this MCCS and, allowing Farmer an advantage due to the different pattern syntax, also the rank of any CCS containing a small \emph{subset} of the embedded predicates.

To artificially embed a pattern, we added $k$ genres, $k$ directors, and $k$
titles to the database, in such a way that each of these $k$ genres and directors are connected to each of the $k$ titles, forming a CCS. As this by itself would create an unrealistic disjoint part of the database, we additionally added random links preserving the overall connectivity and database statistics. E.g., we randomly added links between the existing genres and the newly added titles so as to ensure that, in expectation, the total fraction of titles each of the existing genres
is linked with stays the same. 
This is done also between the existing titles and the newly added genres, and similarly for the directors and titles.

Table~\ref{rank} shows the rank of the embedded MCCS pattern for increasing
$k$. RMiner ranks the embedded pattern higher as the number of nodes per entity type increases and ranks it first when it contains more than just three nodes, showing that RMiner ranks high even relatively small patterns known to be present in the database. 

For Farmer we used the same pattern syntax as in Sec.~\ref{qualitative}. Table~\ref{rank} shows the rank of the highest ranked rule including all \emph{genre} predicates in the embedded MCCS, as well as corresponding to a CCS containing a subset of just two or more of the embedded \emph{genre} predicates. Unsurprisingly, Farmer ranks the CCS patterns more highly than the more specific and thus less frequent MCCS patterns. However, even the CCS patterns are ranked much lower than using RMiner.
\begin{table}[!ht]
\begin{center}
{
\caption{Rank of artificially embedded MCCS pattern in \emph{imdb-3ent-1year} dataset with increasing number of nodes $k$ per entity type.}\label{rank}
\begin{tabular}{c c c c c}
\toprule
$k$ & 2 & 3 & 4 & 6\\
RMiner Rank & 103 & 6 & 1 & 1\\
Farmer Rank (MCCS) & 121 & 502 & 1464 & 2141\\
Farmer Rank (CCS) & 121 & 109 & 125& 147\\
\bottomrule
\end{tabular}
}
\end{center}
\end{table}
\begin{table}
\begin{center}
{
\caption{Time in \lowercase{sec} of the mining step and number of patterns(\#\lowercase{p}) with increasing percentage of `title' nodes(\lowercase{\% n}).}\label{times}
\begin{tabular}{c c c c c c c}
\toprule
&\% n& 20 & 40 & 60 & 80& 100\\
\midrule
RMiner & \#p & 2848 & 5643 & 8575 & 11360& 14202\\
&time &0.6 & 1.9 & 4.1 & 6.1 & 9.1\\
\midrule
Smurfig & \#p & 357954  & 706344 & 1080203 & 1420896 & 1770200\\
&time & 57 & 220 & 561 & 952 & 1451\\
\midrule
Farmer & \#p & 1711 & 2048 & 3424 & 3860& 4545\\
&time &0.02 & 0.07 & 0.13 & 0.14 & 0.19\\
\bottomrule
\end{tabular}
}
\end{center}
\end{table}

\subsection{Scalability comparison}

We did a scalability analysis by running RMiner, Smurfig and Farmer on subsets of \emph{imdb-3ent-1year}, randomly selecting a varying percentage of movies along with the genres and directors connected to these. The results of RMiner show that the number of MCCSs and computation times scale roughly linearly with the database size (see Table~\ref{times}). 
Smurfig is slower than RMiner by a factor of at least 100. Farmer is faster by a factor of 30-40. However, we believe this gap can be shrunk significantly by applying the additional pruning techniques and the use of diffsets, as discussed in Sec.~\ref{algorithm}. 

\section{Conclusion}\label{conclusions}

We have introduced a new syntax of multi-relational patterns in MRDs, and an algorithm to mine them efficiently. Our approach relies on a representation of the MRD as a graph, and mines patterns that correspond to complete connected subgraphs in this graph. This pattern syntax generalizes the notion of a tile in a simple binary database, and is easy to interpret also in more complex settings. Note that while we have written the paper with MRDs in mind, our approach is directly applicable also to RDF data.

An important advantage of the proposed pattern syntax is that it is independent of a notion of support to assess the interestingness of a pattern. Instead, we showed how ideas introduced in \cite{tijl} can be used, defining interestingness by contrasting a pattern with a Maximum Entropy model representing background knowledge on the degree of individual nodes in the graph representation of the MRD.

This paper opens up several avenues for further research, such as: The expansion of the types of prior beliefs that can be taken into account in the MaxEnt model for the database; The definition of different pattern syntaxes corresponding to different graph patterns in the graph representation of the database; The development of speed-ups of the RMiner algorithm e.g. using diffsets and the so-called CHARM-properties presented in \cite{charm}; The development of an algorithm that directly mines the interesting MCCSs, instead of using a two-step approach as in this paper.

\section*{Acknowledgments}
We would like to thank Michael Mampaey for providing the Smurfig code and data and for his support in using Smurfig, as well as Jilles Vreeken for the thorough proof reading and the
insightful feedback. The authors are supported by EPSRC grant EP/G056447/1.

\bibliographystyle{abbrv}
\bibliography{bib}  
\end{document}